\documentclass[12pt]{article}

\usepackage{amssymb,amsfonts,amsmath,amsthm}
\usepackage{graphics}
\usepackage{geometry}
\usepackage{color}
\usepackage{colortbl}

\newtheorem{theorem}{Theorem}
\newtheorem{proposition}[theorem]{Proposition}
\newtheorem{lemma}[theorem]{Lemma}
\definecolor{gris}{gray}{0.75}

\begin{document}

\title {\bf Graphs with the second and third maximum Wiener index over the
 2-vertex connected graphs}
\author{St\'ephane Bessy$^1$, Fran\c cois
 Dross$^2$, Martin Knor$^3$, Riste \v Skrekovski$^4$} \date{} \maketitle
\vspace{-1cm}
\begin{center}
{\small $^1$Laboratoire d'Informatique, de Robotique et de
 Micro\'{e}lectronique de Montpellier\\
 (LIRMM), CNRS, Universit\'e de  Montpellier, France,\\
 \texttt{stephane.bessy@lirmm.fr}\\[3mm]
 $^2$Universit\'e C\^ote d'Azur, I3S, CNRS, Inria, France,\\
 \texttt{francois.dross@inria.fr}\\[3mm]
 $^3$Slovak Technical University in Bratislava,\\
 Faculty of Civil Engineering, Department of Mathematics, Bratislava, Slovakia,\\
 \texttt{knor@math.sk}\\[3mm]
 $^4$Department of Mathematics, University of Ljubljana, Ljubljana, Slovenia,\\
 \texttt{skrekovski@gmail.com}}
\end{center}

\begin{abstract}
Wiener index, defined as the sum of distances between all unordered pairs
of vertices, is one of the most popular molecular descriptors.
It is well known that among 2-vertex connected graphs on $n\ge 3$ vertices,
the cycle $C_n$ attains the maximum value of Wiener index.
We show that the second maximum graph is obtained from $C_n$ by introducing
a new edge that connects two vertices at distance two on the cycle if
$n\ne 6$.
If $n\ge 11$, the third maximum graph is obtained from a $4$-cycle by
connecting opposite vertices by a path of length $n-3$.
We completely describe also the situation for $n\le 10$.
\end{abstract}


\section{Introduction}
The sum of distances between all pairs of vertices in a connected graph
was first introduced by Wiener \cite{Wiener} in 1947.
He observed a correlation between boiling points of paraffins and this
invariant, which has later become known as the Wiener index of a graph.
Today, Wiener index is one of the most used descriptors in chemical graph
theory.

Wiener index was used by chemists decades before it attracted attention
of mathematicians.
In fact, it was studied long before the branch of discrete mathematics,
which is now known as Graph Theory, was developed.
Many years after its introduction, the same quantity has been studied
and referred to by mathematicians as the gross status \cite{Harary},
the distance of graphs \cite{EJS} and the transmission \cite{Soltes}.
A great deal of knowledge on the Wiener index is accumulated in several
survey papers, see e.g. \cite{surv1,KS_chapter, ma5,Gut-sur}.

In what follows, we formally define this index.
Let $d_G(u,v)$ denote the distance between vertices $u$ and $v$ in $G$.
The {\em transmission} of a vertex $v$ is the sum of distances
from $v$ to other vertices of $G$, i.e.,
$w_G(v)=\sum_{u\in V(G)} d_G(u,v)$.
Then the Wiener index of $G$ equals
$$
W(G)=\frac{1}{2}\sum_{u\in G} w_G (u) = \sum_{u,v \in V(G)} d_G(u,v).
$$

Due to big importance and popularity, there are many results about graphs
with extremal (either maximum or minimum) values of Wiener index in
particular classes, see the surveys mentioned above.
However, only few papers are devoted to the second, third, etc extremal
graphs, although it is important to understand the ordering of graphs by
Wiener index.
One of the reasons is that results of this type are much more complicated,
often including the extremal graph as a trivial case.
Of course, the situation is known for trees.
In \cite{DG} there are described the first 15 trees with the smallest value of
Wiener index.
Analogously, in \cite{Deng,LLL} there are the first 15 trees with the greatest
value of Wiener index.
Graphs with the second minimum and second maximum value of Wiener index over
the class of unicyclic graphs are found in \cite{FIY}.
In this paper we describe graphs with the second and third maximum value of
Wiener index over the class of 2-vertex connected graphs.

We use the following notation.
As usual, $C_n$ is the cycle on $n$ vertices.
Let $H_{n,p,q}$ be a graph on $n$ vertices comprised of three internally
disjoint paths with the same end-vertices, where the first one has length
$p$, the second one has length $q$, and the last one has length $n-p-q+1$.
Notice that $H_{n,p,q}$ has $n$ vertices.
% and is isomorphic to the theta graph $\Theta_{p,q,n-p-q+1}$.
Also observe that $H_{n,1,2}$ is the cycle on $n$ vertices plus an edge
linking two vertices at distance two on the cycle.
When using the notation $H_{n,p,q}$ we assume that
$1\le p\le q\le n-p-q+1$ and $q>1$.
Our main result is the following theorem.

\begin{theorem}
\label{thm:main}
Let $n\ge 11$ and let $G$ be a 2-vertex connected graph on $n$ vertices
different from $C_n$, $H_{n,1,2}$ and $H_{n,2,2}$.
Then
$$
W(G)<W(H_{n,2,2}) < W(H_{n,1,2}) < W(C_{n}).
$$
\end{theorem}

%
% --------------------------------  2-verte-connected -----------------------
%

\section{Proof of the result}

We start with some definitions.
For two vectors $\omega$ and $\omega'$ of the same finite dimension,
we write $\omega \preceq \omega'$ if for every coordinate $i$ we have
$\omega_i\le \omega'_i$.
Moreover we define $\langle \omega \rangle$ as the value $\sum_{i} i \omega_i$.
It is clear that $\omega \preceq \omega'$ implies
$\langle \omega \rangle \le \langle \omega' \rangle$.

Let $G$ be a connected graph on $n$ vertices and let $v$ be a vertex of $G$.
The {\em distance vector} of $v$ is the $(n-1)$-dimensional vector
$\omega_G(v)$ given by $\omega_G(v)_i=|\{x\in G\ : \ d_G(v,x)=i\}|$.
Observe that $\langle \omega_G(v) \rangle=w_G(v)$.

If $n$ is even, the vector ${\bf 2}_n$ has dimension $n/2$ and contains
the value 2 in each coordinate except for the last one which is 1.
If $n$ is odd, ${\bf 2}_n$ has dimension $(n-1)/2$ and each of its
coordinates has value 2.
For example ${\bf 2}_6=(2,2,1)$ and ${\bf 2}_7=(2,2,2)$.

Let $G$ be a 2-vertex connected graph and let $v$ be a vertex of $G$.
Since $G$ has no cut-vertices, every coordinate of $\omega_G(v)$
has value at least 2, except for the last one which can be 1.
In other words, for every vertex $v$ of a 2-vertex connected graph $G$
we have $\langle \omega_G(v) \rangle \le \langle {\bf 2}_n \rangle =
\lfloor \frac{n^2}{4} \rfloor$.
This implies the following classical result.

\begin{theorem}
\label{thm:2vconnected}
For every $n\ge 3$, the cycle $C_n$ is the unique graph which has
the maximum Wiener index over the class of 2-vertex connected graphs on
$n$ vertices.
Moreover, $W(C_n)=\frac{n}{2}\langle {\bf 2}_n\rangle$.
\end{theorem}

Now we describe the structure of graphs with the second and third
maximum Wiener index over the class of $2$-vertex connected graphs
on $n$ vertices.
First we need some definitions and lemmas.

We denote by $k(v)$ the first coordinate $i$ of $\omega_G(v)$ such that
$\omega_G(v)_i>2$.
If such a coordinate does not exist, we set
$k(v)=\lfloor\frac{n}{2}\rfloor$.
Notice that if $\omega_G(v)\ne{\bf 2}_n$, then
$k(v) < \lfloor\frac{n}{2}\rfloor$.
For a graph $G$ on $n$ vertices we denote by $k(G)=(k_i(G))_{1\le i\le n}$
the sequence formed by the values $k(v)$ of all $v\in V(G)$ given
in non-decreasing order.
For instance, the sequence $k(H_{n,1,2})$ is given by
$k_i(H_{n,1,2})=\lfloor \frac{i+1}{2} \rfloor$ for every $i=1,\dots ,n-1$
and $k_n(H_{n,1,2})=\lfloor \frac{n}{2} \rfloor$.
In other words we have
$k(H_{n,1,2})=(1,1,2,2,3,3,\dots,\lfloor\frac{n}{2} \rfloor)$ with
twice the value $\lfloor \frac{n}{2} \rfloor$ at the end if $n$ is
even and three times if $n$ is odd.
Similarly for $n\ge 5$, the sequence $k(H_{n,2,2})$ is given by
$k_1(H_{n,2,2})=k_2(H_{n,2,2})=1$, $k_3(H_{n,2,2})=k_4(H_{n,2,2})=2$
and $k_i(H_{n,2,2})=\lfloor \frac{i-1}{2} \rfloor$ for every $i=5,\dots,n$.
In other words we have
$k(H_{n,2,2})=(1,1,2,2,2,2,3,3,4,4,\dots,\lfloor\frac{n-1}{2}\rfloor)$
with once the value $\lfloor\frac{n-1}{2} \rfloor$ if $n$ is
odd and twice if it is even.

As previously precised, we write $k(G)\preceq k(G')$ if for every
$i\in \{1,\dots ,n\}$ we have $k_i(G)\le k_i(G')$.
Moreover, if $k_j(G)< k_j(G')$ for some $j\in \{1,\dots ,n\}$, then we
write $k(G)\prec k(G')$ .

The next two lemmas give necessary conditions to bound the Wiener
index of a graph by the ones of $H_{n,2,2}$ and $H_{n,1,2}$.

\begin{lemma}
\label{lem:kimpliesW}
Let $G$ be a 2-vertex connected graph on $n\ge 5$ vertices.
If $k(G)\prec k(H_{n,1,2})$, then $W(G)<W(H_{n,1,2})$.
Similarly if $k(G)\prec k(H_{n,2,2})$, then $W(G)<W(H_{n,2,2})$.
\end{lemma}

\begin{proof}
Let $H$ be one of the graphs $H_{n,1,2}$ or $H_{n,2,2}$, and let $h\in V(H)$
with $k(h)<\lfloor \frac{n}{2}\rfloor$.
The vector $\omega_{H}(h)$ has 3 at coordinate $k(h)$ and 2 everywhere else,
except possibly for the last coordinate.
Therefore, $\langle \omega_{H}(h) \rangle$ has the largest value
among
\begin{equation}
\label{eq:w}
\big\{\langle \omega_{G'}(u) \rangle \, :\  u\in V(G')\textrm{ and }
\ k(u)\le k(h)\big\},
\end{equation}
where $G'$ is a 2-vertex connected graph of order $n$.
The same conclusion also holds if $k(h)=\lfloor \frac{n}{2}\rfloor$ as
in this case $\omega_{H}(h)={\bf 2}_n$.

%% Now assume that $k(G)\prec k(H)$.
%% Relabel the vertices of $G$ (resp.~$H$) so that $V(G)=\{v_1,\dots ,v_n\}$
%% (resp.~$V(H)=\{v'_1,\dots ,v'_n\}$) and $k(v_i)=k_i(G)$
%% (resp.~$k(v'_i)=k_i(H)$) for $i=1,\dots,n$.
%% Let $f$ be a bijection from $V(G)$ to $V(H)$ such that $f(v_i)=v'_i$
%% for $i=1,\dots ,n$.
%% By assumption, for every $x\in V(G)$ we have $k(x)\le k(f(x))$,
%% and there exists $y\in V(G)$ such that $k(y)< k(f(y))$.
%% Further, for every $x\in V(G)$ we have
%% $\langle \omega_{G}(x) \rangle \le \langle \omega_{H}(f(x))\rangle$, by
%% (\ref{eq:w}).
%% Since we also have
%% $\langle\omega_{G}(y)\rangle < \langle\omega_{H}(f(y))\rangle$, we obtain
%% $2\cdot W(G)<2\cdot W(H)$ and conequently $W(G)<W(H)$.

Now assume that $k(G)\prec k(H)$.
Relabel the vertices of $G$ (resp.~$H$) so that $V(G)=\{v_1,\dots ,v_n\}$
(resp.~$V(H)=\{v'_1,\dots ,v'_n\}$) and $k(v_i)=k_i(G)$
(resp.~$k(v'_i)=k_i(H)$) for $i=1,\dots,n$.
By assumption, for every $i=1,\dots ,n$ we have $k(v_i)\le k(v'_i)$,
and there exists $i_0 \in \{1,\dots ,n\}$ such that $k(v_{i_0})< k(v'_{i_0})$.
Further, for every $i=1,\dots ,n$ we have
$\langle \omega_{G}(v_i) \rangle \le \langle \omega_{H}(v'_i)\rangle$, by
(\ref{eq:w}).
Since we also have
$\langle\omega_{G}(v_{i_0})\rangle < \langle\omega_{H}(v'_{i_0})\rangle$, we obtain
$2\cdot W(G)<2\cdot W(H)$ and consequently $W(G)<W(H)$.
\end{proof}

In particular, since for $n\ge 6$ we have $k(H_{n,2,2})\prec k(H_{n,1,2})$,
we obtain $W(H_{n,2,2})<W(H_{n,1,2})$.

For a graph $G$ and $x \in V(G)$, we say that $x$ is \emph{bad} if
$\omega_G(x)$ has at least two elements with value greater than 2.
For every bad vertex $x$ in $G$, let $k'(x)$ be the coordinate of the
second element which is at least 3 in $\omega_G(x)$.
For all bad vertices $x$ of $G$ we sum the values
$\lfloor\frac{n-1}{2}\rfloor - k'(x)$, and we denote by $b(G)$ the result.

\begin{lemma}
\label{lem:kimpliesWbis}
Let $H$ be either $H_{n,1,2}$ or $H_{n,2,2}$.
Let $G$ be a 2-vertex-connected graph on $n\ge 4$ vertices such that
$\sum_{x\in V(G)} k(x) < b(G) + \sum_{x \in V(H)}k(x)$.
Then $W(G)<W(H)$.
\end{lemma}

\begin{proof}
Let $i \in \{1,\dots,\lfloor \frac{n}{2} \rfloor\}$ and let $v$
be a vertex in $H$ with $k(v) = i$.
The vertex $v$ has $3$ at coordinate $i$ and $2$ everywhere else,
except for the last coordinate that may be $1$.
Such a vertex $v$ satisfies
$\langle\omega(v)\rangle= \langle\omega({\bf 2}_{n-1})\rangle+i$.
Moreover, as noticed in the proof of Lemma~\ref{lem:kimpliesW},
for every $u$ in $G$ with $k(u)=i$ we have
$\langle\omega(v)\rangle \ge \langle\omega(u)\rangle$, see (\ref{eq:w}).

Similarly, for all $j \in\{i,\dots,\lfloor \frac{n-1}{2} \rfloor\}$,
among the bad vertices $v$ with $k(v) = i$ and $k'(v) = j$, the
highest possible value $\langle\omega(v)\rangle$ is obtained when
$\omega(v)$ has $3$ at coordinates $i$ and $j$, and $2$ at each other
coordinate (except possibly the last one).
Thus, for every bad vertex $v$ with $k(v) = i$
and $k'(v) = j$, we have
$\langle\omega(v)\rangle \le \langle\omega({\bf 2}_{n-2})\rangle+i+j =
\langle\omega({\bf 2}_{n-1})\rangle - \lfloor \frac{n-1}{2} \rfloor +i+j$.
So we have
$$
2\cdot W(G) \le \sum_{x \in V(G)}(\langle\omega({\bf 2}_{n-1})\rangle+ k(x)) - b(G)
$$ 
and 
$$
2\cdot W(H) = \sum_{x \in V(H)}(\langle\omega({\bf 2}_{n-1})\rangle+ k(x)).
$$
Hence, the lemma easily follows.
\end{proof}

In the next two propositions we consider two particular sub-classes of
the 2-vertex connected graphs.
For $n\ge 6$, let ${\cal H}_n$ be the class of graphs comprised of
$H_{n,1,q}$ for $q=3,\dots, \lfloor\frac{n}{2}\rfloor$.
We have the following claim.

\begin{proposition}
\label{prop:H}
Let $n = 9$ or $n \ge 11$, and let $G$ be a graph of ${\cal H}_n$.
Then
$$
W(G)<W(H_{n,2,2})<W(H_{n,1,2}).
$$
For $n\le 8$ and $n=10$, the values of $W(G)$ for
$G\in {\cal H}_n\cup\{H_{n,1,2},H_{n,2,2}\}$ are summarised in the table
below.
\end{proposition}

\scriptsize
\begin{center}
\begin{tabular}{r||c|c|c|c}
$n$ & 4 & 5 & 6 & 7 \\
\hline
& & & & \\
${\cal H}_n$ & $H_{4,1,2}$ &  $H_{5,1,2}$~~$H_{5,2,2}$ &
$H_{6,1,2}$~~$H_{6,1,3}$~~$H_{6,2,2}$ & $H_{7,1,2}$~~$H_{7,1,3}$~~$H_{7,2,2}$ \\
& & & & \\
$W$ & 7 & 14~~~~~~14 & 24~~~~~~25~~~~~~23 & 39~~~~~~38~~~~~~38 \\
\end{tabular}\phantom{v}\\
\vskip 2mm
\begin{tabular}{r||c|c}
$n$ & 8 & 10 \\
\hline
& & \\
${\cal H}_n$ & $H_{8,1,2}$~~$H_{8,1,3}$~~$H_{8,1,4}$~~$H_{8,2,2}$
& $H_{10,1,2}$~~$H_{10,1,3}$~~$H_{10,1,4}$~~$H_{10,1,5}$~~$H_{10,2,2}$\\
& & \\
$W$ & 58~~~~~~58~~~~~~55~~~~~~56 & 115~~~~~~113~~~~~~107~~~~~~109~~~~~~112 \\
\end{tabular}
\end{center}
\normalsize
\begin{center}
Table~1: Values of $W(H_{n,p,q})$ for $n\le 8$ and $n=10$.
\end{center}

\begin{proof}
Assume that $n \ge 9$.
As noticed below Lemma~\ref{lem:kimpliesW}, we have $W(H_{n,2,2})<W(H_{n,1,2})$.
Let us now focus on $H_{n,1,q}$ for $3\le q\le n-q$.
Let $u$ and $v$ be the two vertices of degree 3 in $H_{n,1,q}$, and let
$x_1, \dots, x_{q-1}$ be the internal vertices of the path of length $q$
from $u$ to $v$.
Further, let $y_1,\dots ,y_{n-q-1}$ be the internal vertices of the path of length
$n-q$ from $u$ to $v$.
We have $k(u)=k(v)=1$.
For $i=1,\dots,\lfloor \frac{q}{2}\rfloor -1$ we have $k(x_i)=k(x_{q-i})=i+1$
and for $i=1,\dots,\lfloor \frac{n-q}{2} \rfloor-1$ we have
$k(y_i)=k(y_{n-q-i})=i+1$.
Finally, if $q$ is odd then we have
$k(x_{(q-1)/2}) = k(x_{(q+1)/2}) = \lfloor \frac{n}{2} \rfloor$ and if $q$ is
even then we have $k(x_{q/2}) = \lfloor \frac{n}{2} \rfloor$.
Similarly if $n-q$ is odd then we have
$k(y_{(n-q-1)/2}) = k(y_{(n-q+1)/2}) = \lfloor \frac{n}{2} \rfloor$ and if $n-q$
is even then we have $k(y_{(n-q)/2}) = \lfloor \frac{n}{2} \rfloor$.
Therefore the numbers in $k(H_{n,1,q})$ are at least twice each integer from $1$
to $\lfloor \frac{n-q}{2} \rfloor$, plus at most three times $\frac{n-1}{2}$ if
$n$ is odd and at most four times $\frac{n}{2}$ if $n$ is even.
Therefore we have $\sum_{x \in V(H_{n,1,q})} k(x) \le
3\cdot\frac{n-1}{2} + 2\cdot\sum_{i \in \{1,..., \frac{n-3}{2}\}}i$
if $n$ is odd and $\sum_{x \in V(H_{n,1,q})} k(x) \le
4\cdot\frac{n}{2} + 2\cdot\sum_{i \in \{1,...,\frac{n-4}{2}\}}i$ if $n$ is
even.

Recall that $k(H_{n,2,2}) = (1,1,2,2,2,2,3,3,4,4,\dots,\lfloor
\frac{n-1}{2}\rfloor)$, with once the value $\lfloor \frac{n-1}{2}\rfloor$
if $n$ is odd and twice if $n$ is even.
Therefore if $n$ is odd, then
$\sum_{x \in V(H_{n,1,q})} k(x) - \sum_{x \in V(H_{n,2,2})}k(x) \le
2(\frac{n-1}{2} - 2) = n-5$; and if $n$ is even, then $\sum_{x \in
V(H_{n,1,q})} k(x) - \sum_{x\in V(H_{n,2,2})}k(x) \le 2(\frac{n}{2}-2)+2
=n-2$.

Now notice that since $q \ge 3$ and $n - q \ge 4$, $u$ and $v$ are bad
vertices with $k' = 2$.
For instance, $u$ has three distinct vertices at distance two,
which are $x_2$, $y_2$ and $y_{n-q-1}$.
Moreover, since $n \ge 9$, $y_1$ and $y_{n-q-1}$ are bad vertices with
$k' = 3$.
For instance, if $q = 3$ then $n-q-1\ge 5$ and $x_2$, $y_4$, and $y_{n-q-1}$ are
distinct vertices at distance $3$ from $y_1$, and if $q \ge 4$, then
$x_2$, $x_{q-1}$, and $y_4$ are distinct vertices at distance $3$ from
$y_1$.
Moreover, if $n\ge 11$ we can find another bad vertex.
Indeed, if $q\ge 5$, then $n-q\ge q\ge 5$ and $x_4$, $y_2$
and $y_{n-q-1}$ are distinct vertices at distance 3 from $x_1$,
which is bad then.
If $q=4$, then $n-q\ge 7$ and $x_2$, $x_{q-1}$ and $y_6$ are at distance
4 from $y_2$, which is bad.
And if $q=3$, then $n-q\ge 8$ and $x_2$, $y_6$ and $y_{n-q-1}$ are
at distance 4 from $y_2$ which is bad.
Therefore $b(H_{n,1,q}) \ge 4 \cdot \lfloor\frac{n-1}{2}\rfloor
- (2\cdot 2 + 2\cdot 3)$, with a strict inequality if $n \ge 11$.
If $n = 9$ then $b(H_{n,1,q}) \ge 6$ and
$\sum_{x \in V(H_{n,1,q})} k(x) - \sum_{x \in V(H_{n,2,2})}k(x) \le
4$. If $n =11$ then $b(H_{n,1,q}) \ge 10$ and $\sum_{x \in
V(H_{n,1,q})} k(x) - \sum_{x \in V(H_{n,2,2})}k(x) \le 6$.
If $n \ge 12$, then $b(H_{n,1,q}) > 2n - 14 \ge n - 2$ and $\sum_{x \in
V(H_{n,1,q})} k(x) - \sum_{x \in V(H_{n,2,2})}k(x) \le n - 2$.
In all cases we have $W(H_{n,1,q})<W(H_{n,2,2})$, by Lemma~\ref{lem:kimpliesWbis}.
\end{proof}

For $n\ge 4$, let ${\cal I}_n$ be the class of graphs built from $C_n$
by adding two distinct edges, each linking two vertices at distance
precisely 2 along $C_n$.
That is, a graph $G$ belongs to ${\cal I}_n$ if $V(G)=\{x_1,\dots ,x_n\}$
and $E(G)=\{x_ix_{i+1}\ :\ i=1,\dots ,n-1\}\cup\{x_nx_1,x_1x_3,x_ix_{i+2}\}$,
where $1<i\le n-2$.
Further, by $G_6^1$ we denote a graph from ${\cal I}_6$ when $i=4$.
So $G_6^1$ consists of two disjoint triangles connected by two independent
edges.
We have the following claim.

\begin{proposition}
\label{prop:I} 
Let $n\ge 5$.
Every graph $G$ of ${\cal I}_n$ satisfies the inequality
$W(G)<W(H_{n,2,2})$, with the unique exception of $G_6^1$ for which
$W(G_6^1)=W(H_{6,2,2})$.
\end{proposition}

\begin{proof}
Let $G\in{\cal I}_n$.
If $i=2$ then $G$ is a strict supergraph of $H_{n,2,2}$, which implies
$W(G)<W(H_{n,2,2})$.

So assume that $i\ne 2$.
We compare $W(G)$ with $W(H_{n,1,3})$.
Notice that both $G$ and $H_{n,1,3}$ are cycles of length $n-2$ with two
additional vertices.
Let us call these two additional vertices by $u_1$ and $u_2$ in $G$,
and by $u'_1$ and $u'_2$ in $H_{n,1,3}$.
Then
\begin{align}
W(H_{n,1,3}) - W(G)=&\ d(u'_1,u'_2)
+ w_{H_{n,1,3} - u'_2}(u'_1) + w_{H_{n,1,3} - u'_1}(u'_2) \nonumber\\
-&\ d(u_1,u_2) - w_{G - u_2}(u_1) - w_{G - u_1}(u_2). \nonumber
\end{align}
Since $G$ is $2$-connected, we have
$d(u_1,u_2) \le \lfloor \frac{n}{2} \rfloor$, and in $H_{n,1,3}$,
$u'_1$ and $u'_2$ are at distance $1$.
Further, $G - u_2$ is a cycle of length $n-2$ with an additional
vertex $u_1$ adjacent to two neighbours in the cycle, and
$H_{n,1,3} - u'_2$ is a cycle of length $n-2$ with an additional vertex
$u'_1$ adjacent to one vertex in the cycle.
Therefore $w_{H_{n,1,3} - u'_2}(u'_1) - w_{G - u_2}(u_1)
= \lfloor \frac{n-2}{2} \rfloor$, and by symmetry, $w_{H_{n,1,3} -
u'_1}(u'_2) - w_{G - u_1}(u_2) = \lfloor \frac{n-2}{2} \rfloor$.
Therefore, $W(H_{n,1,3})-W(G) \ge \lfloor\frac{n-2}{2}\rfloor$.
Now combining this inequality with Proposition~{\ref{prop:H}}
we obtain that $W(H_{n,2,2})-W(G)>0$ for all $n\ge 5$, with a unique
exception when $n=6$ and $W(H_{6,2,2})-W(G)=0$.
In this case we must have $d(u_1,u_2)=\frac n2$, and consequently $G=G_6^1$.
\end{proof}

Consider $H_{n,1,3}$ with vertices $\{x_1,\dots,x_n\}$ and edges
$\{x_ix_{i+1}\ : \ i=1,\dots,n{-}1\}\cup\{x_nx_1,x_1x_4\}$.
Let $G_6^2$ be obtained from $H_{6,1,3}$ by adding the edge $x_1x_3$,
and let $G_6^3$ be obtained from $H_{n,1,3}$ by adding the edge $x_2x_5$.
Denote ${\cal G}=\{G_6^1,G_6^2,G_6^3\}$.
Moreover, let $G_8^1$ be obtained from $H_{8,1,3}$ by adding the edge
$x_5x_8$.
Observe that $W(G)=W(H_{6,2,2})$ when $G\in{\cal G}$, and
$W(G_8^1)=W(H_{8,2,2})$.
The following theorem implies and precises Theorem~\ref{thm:main}.

\begin{theorem}
\label{thm:2vconnectedbis}
For $n =4$, there are three $2$-vertex connected graphs and they satisfy
$W(K_4)<W(H_{4,1,2})<W(C_4)$.
For every $n\ge 5$, let $G$ be a 2-vertex connected graph on $n$ vertices
different from $C_n$, $H_{n,1,2}$, $H_{n,2,2}$ and $H_{n,1,3}$.
Moreover, assume $G\notin{\cal G}$ if $n=6$ and $G\ne G_8^1$ if $n=8$.
We have :
\begin{itemize} 
\item $W(G) < W(H_{n,2,2})=W(H_{n,1,2})$ for $n = 5$,
\item $W(G) < W(H_{n,2,2}) < W(H_{n,1,2}) < W(H_{n,1,3})$ for $n = 6$,
\item $W(G) < W(H_{n,2,2}) = W(H_{n,1,3}) < W(H_{n,1,2})$ for $n = 7$,
\item $W(G) < W(H_{n,2,2}) < W(H_{n,1,3}) = W(H_{n,1,2})$ for $n = 8$,
\item $W(G) < W(H_{n,2,2}) < W(H_{n,1,3}) < W(H_{n,1,2})$ for $n = 10$, and
\item $W(G)<W(H_{n,2,2}) < W(H_{n,1,2})$ and $W(H_{n,1,3}) <
  W(H_{n,2,2})$ for $n = 9$ and $n \ge 11$.
\end{itemize}
\end{theorem}

\begin{proof}
For $n\ge 5$, let ${\cal C}$ be the class of $2$-vertex connected
graphs on $n$ vertices different from $C_n$, $H_{n,1,2}$, $H_{n,2,2}$ and
$H_{n,1,3}$.
Let $G$ be a graph with the maximum Wiener index over ${\cal C}$.
We want to prove $W(G)<W(H_{n,2,2})$ except when $G \in {\cal G}$ or $G = G_8^1$.
Notice that no proper subgraph of $G$ is in ${\cal C}$, since otherwise this
proper subgraph would have a bigger Wiener index than $G$.
%This means also that $G$ cannot contain a $4$-cycle with a chord. \textcolor{red}{That is not true.}
First suppose that $G$ has a Hamiltonian cycle $C$.
We distinguish three cases.

{\bf Case 1:}
{\it $G$ contains an edge $xy$ where $x$ and $y$ are at distance at least 4
along $C$.}
Then $C+xy$ itself is a graph from ${\cal C}$.
Thus, by the choice of $G$ we have $G=C+xy$ and $G=H_{n,1,q}$ for some
$q\ge 4$.
By Proposition~\ref{prop:H}, $W(G)<W(H_{n,2,2})$.

{\bf Case 2:}
{\it $G$ contains an edge $xy$ where $x$ and $y$ are at distance 3 along
$C$.}
Since the Hamiltonian cycle $C$ with $xy$ is $H_{n,1,3}$, $G$ must contain
one more edge, say $st$.
Since $H_{n,1,3}$ with $st$ is in $\cal C$, there are no other edges in $G$.
By Case~1 we may assume that $s$ and $t$ are at distance 2 or 3 along $C$.
Since $G$ is a supergraph of $H_{n,1,3}$, we have $W(G)<W(H_{n,2,2})$ if
$n\notin\{6,8,10\}$, by Proposition~{\ref{prop:H}}.
The cases when $n\in\{6,8,10\}$ were checked by a computer and it was found
that $W(G)<W(H_{n,2,2})$ with two exceptions if $n=6$, namely when
$G=G_6^2$ and $G=G_6^3$, and with one exception if $n=8$, namely when
$G=G_8^1$.

{\bf Case 3:}
{\it The edges of $G$ not belonging to $C$ link vertices at distance 2
along $C$.}
Let us denote by $e_1,\dots ,e_\ell$ these edges.
Since $C+e_1+e_2$ itself is a graph from $\cal C$, we have $G=C+e_1+e_2$.
Now Proposition~\ref{prop:I} concludes the proof.

So assume that $G$ has no Hamiltonian cycle.
Let $v$ be a vertex of $G$ with the maximum value of $k$, that is
$k(v)=k_n(G)$.
We denote this value by $p$.
If $p=1$, it is clear that $k(G)\prec k(H_{n,2,2})$ and
Lemma~\ref{lem:kimpliesW} implies the result.
So assume that $p\ge 2$.
We know that there are exactly two vertices at distance $i$ from $v$
for every $i\in \{1,\ldots,p-1\}$.
Denote these vertices by $u_i$ and $v_i$.
Notice that for $i=1,\dots,p-2$ the only neighbours of $u_i$ and $v_i$
are contained in $\{u_{i-1},v_{i-1},u_i,v_i,u_{i+1},v_{i+1}\}$
(with $u_0=v_0=v$).
Moreover, since $G$ is 2-vertex connected, there exists a matching
of size 2 between $\{u_i,v_i\}$ and $\{u_{i+1},v_{i+1}\}$ for
$i=1,\dots,p-2$.
So we assume that $u_iu_{i+1}$ and $v_iv_{i+1}$ are edges of $G$ for
$i=1,\dots,p-2$ and thus that $P=u_{p-1},\dots,u_1,v,v_1,\dots ,v_{p-1}$
is a path of $G$.
Finally, denote by $X$ the set $(N_G(u_{p-1})\cup N_G(v_{p-1}))\setminus
\{u_{p-2},v_{p-2},u_{p-1},v_{p-1}\}$.
Let $G'$ be the subgraph of $G$ obtained by removing the edges
of $G[V(P)]$ which do not belong to $P$.
Notice first that $G'$ is a 2-vertex connected graph.
Indeed, since $G'\setminus \{u_{p-2},u_{p-3},\dots,v_{p-3},v_{p-2}\}$
is connected (otherwise $u_{p-1}$ or $v_{p-1}$ would be a cut-vertex of
$G$), no vertex of $P$ is a cut-vertex of $G'$.
Moreover, no vertex of $G'\setminus P$ is a cut-vertex of $G'$ otherwise
it would be a cut-vertex of $G$.
Furthermore, $G'$ is not a cycle or $H_{n,1,2}$ or $H_{n,1,3}$, otherwise
$G$ would have a Hamiltonian cycle.
We may also assume that $G'$ is not $H_{n,2,2}$, since otherwise $G$ is a
supergraph of $H_{n,2,2}$ and $W(G)<W(H_{n,2,2})$.
Hence, $G'$ belongs to $\cal C$, and by the choice of $G$ we have $G'=G$.
We consider two cases.

{\bf Case 1:}
{\it $p=\lfloor \frac{n}{2} \rfloor$.}
In this case $|X|=1$ or $|X|=2$.
If $|X|=1$, denote by $x$ the unique vertex of $X$.
Since $G$ is 2-vertex connected, $u_{p-1}x$ and $v_{p-1}x$ are edges
of $G$ and $G$ has a Hamiltonian cycle, contradicting a previous assumption.
If $|X|=2$, denote by $u_p$ and $v_p$ the vertices of $X$.
Analogously as above, since $G$ is 2-vertex connected, we can assume that
$u_{p-1}u_p$ and $v_{p-1}v_p$ are edges of $G$.
Since $G$ has no Hamiltonian cycle, $u_pv_p$ is not an edge of $G$.
But $G$ is 2-vertex connected, and so $u_{p-1}v_p$ and $v_{p-1}u_p$
are edges of $G$.
Hence $G=H_{n,2,2}$.

{\bf Case 2:}
{\it $p<\lfloor \frac{n}{2} \rfloor$.}
Below we will show that $k(G)$ admits a non-decreasing subsequence
${\kappa}=(k_1,\dots,k_{2p+1})$ with
${\kappa}\prec(1,1,2,2,2,2,3,3,4,4,\dots,p-1,p-1,p)$, and the existence of
a coordonate $i_0$ for which $k_{i_0}(G)<k_{i_0}(H_{n,2,2})$.
Then we will conclude that $W(G)<W(H_{n,2,2})$.
Indeed, for every value $k_i(G)>2$ with $i<n$ there will exist at least
$2k_i(G)$ elements before it in $k(G)$, which means that $i\ge 2k_i(G)+1$,
and hence $k_i(G)\le \lfloor\frac{i-1}{2}\rfloor=k_i(H_{n,2,2})$.
If $k_i(G)=1$, then we have $k_i(G)\le k_i(H_{n,2,2})$.
Further, $G$ has at least two vertices of degree at least 3, for otherwise
$G$ would not be 2-vertex connected.
Hence, if $k_i(G)=2$, then we have $i\ge 3$ and $k_i(G)\le k_i(H_{n,2,2})$.
Moreover, as $k_{i_0}(G)<k_{i_0}(H_{n,2,2})$ we have
$k(G)\prec k(H_{n,2,2})$.
By Lemma~\ref{lem:kimpliesW}, we conclude that $W(G)<W(H_{n,2,2})$.

Thus, all that remains to show is the existence of subsequence $\kappa$ of $k(G)$ with
${\kappa}\prec(1,1,2,2,2,2,3,3,$ $4,4,\dots,p-1,p-1,p)$, and the existence of
a coordonate $i_0$ for which $k_{i_0}(G)<k_{i_0}(H_{n,2,2})$.
Since $p<\lfloor\frac{n}{2}\rfloor$ we have $|X|\ge 3$, and we may assume
that $u_{p-1}$ has at least two neighbours in $X$.
We find a special path $Q$ in $G$.
There are two cases to consider.
First, if $v_{p-1}$ has at least two neighbours in $X$,
then $k(u_i)=k(v_i)=p-i$ for every $i=1,\dots,p-1$.
In this case we set $x=u_{p-1}$, $y=v_{p-1}$ and $Q=P$.
So $Q$ is an induced path, the only neighbours of $Q$ in $G\setminus Q$
are those of $x$ and $y$ and $\kappa_Q=(1,1,2,2,\dots ,p-1,p-1,p)$ is
a subsequence of $k(G)$ achieved only by vertices of $Q$.
If $v_{p-1}$ has only one neighbour in $X$, then $v_{p-1}$ has degree 2,
and we denote by $v_p$ its neighbour different from $v_{p-2}$.
The degree of $v_p$ is at least 3, for otherwise we would have
$k(v_1)\ge p+1$, which contradicts the fact that $v$ has the maximum value
in $k(G)$.
So we have $k(v_1)=p$ and $k(u_i)=k(v_{i+1})=p-i$ for $i=1,\dots,p-1$.
In this case we set $x=u_{p-1}$, $y=v_{p}$ and $Q=P\cup \{v_p\}$.
If $Q$ is not an induced path in $G$, then $u_{p-1}v_p=xy$ is an
edge of $G$.
But since $v_p$ is not a cut-vertex of $G$, $G-xy$ is
2-vertex connected.
Since $G$ has no Hamiltonian cycle, $G-xy$ is different from $C_n$,
$H_{n,1,2}$ and $H_{n,1,3}$.
And analogously as before Case~1 we may assume that $G-xy$ is different from
$H_{n,2,2}$.
So $G-xy\in {\cal C}$ which is not possible.
Thus here again, $Q$ is an induced path in $G$, the only neighbours
of $Q$ in $G\setminus Q$ are those of $x$ and $y$ and
$\kappa_Q=(1,1,2,2,\dots ,p-1,p-1,p,p)$ is a subsequence of $k(G)$
achieved only by vertices of $Q$.
To conclude the proof, we analyse three different cases.

First assume that $G$ contains a vertex $z$ with degree at least 4.
If $z\in \{x,y\}$ then $z$ has at least three neighbours in $G\setminus Q$.
On the other hand, if $z\notin \{x,y\}$ then $z$ and at least two its
neighbours are in $G\setminus Q$.
In any case, at least three vertices of $\{z\}\cup N(z)$ are in
$G\setminus Q$.
We show that all these vertices have $k$ at most $2$.
Obviously $k(z)=1$.
So let $z_1$ be a neighbour of $z$ outside $Q$.
Suppose that $k(z_1)>2$.
If the degree of $z_1$ is at least 3, then $k(z_1)=1$, a contradiction.
Therefore $z_1$ has exactly two neighbours.
Observe that $d(z_1,s)\le 2$ whenever $s$ is $z$ or one of its neighbours.
Since $\{z\}\cup N(z)$ has at least five vertices, the other neighbour of
$z_1$ (different from $z$) must be in $N(z)$.
Denote this neighbour by $z_2$.
Since $z$ is not a cut-vertex, $z_2$ has a neighbour, say $q$, which is
different from $z$ and $z_1$.
If $q\in N(z)$ then $z,z_1,z_2,q$ form a 4-cycle with a chord in $G$,
a contradiction.
On the other hand if $q\notin N(z)$ then $k(z_1)=2$ which contradicts our
assumption that $k(z_1)>2$.
Hence $k(z_1)\le 2$, and the same holds for all neighbours of $z$ outside
$Q$.
It means that the three vertices of $\{z\}\cup N(z)$ outside $Q$ together
with $\kappa_Q$ yield a sequence, first $2p+1$ members of which form
the desired sequence $\kappa$.
Moreover, we know that $k_7(G)$ is at most the seventh value of $\kappa$,
and so $k_7(G)\le 2 < k_7(H_{n,2,2})=3$.
Consequently, $W(G)<W(H_{n,2,2})$.
Therefore, in the next we assume that every vertex of $G$ has degree at most
3.

Now suppose that $G$ contains at least four vertices $z_1$, $z_2$, $z_3$
and $z_4$ of degree 3.
Two of these vertices at least, say $z_1$ and $z_2$, do not belong to $Q$,
and so $\{k(z_1),k(z_2)\}\cup \kappa_Q$ contains the desired sequence
$\kappa=(1,1,1,1,2,2,3,3,\dots , p-1,p-1,p)$.
Since $k_3(G)=1 < k_3(H_{n,2,2})=2$, we have $W(G)<W(H_{n,2,2})$.

Finally, suppose that $G$ has exactly two vertices of degree 3,
while all the other vertices have degree 2.
So $G$ is $H_{n,a,b}$, and $x,y$ are connected by paths of length $a$, $b$
and $n-a-b+1$, where the last one is $Q$.
Since $G$ has no Hamiltonian cycle, $a\ge 2$.
And since $G$ is different from $H_{n,2,2}$, we have $b\ge 3$.
Then $G\setminus Q$ has at least three vertices, say $z_1$, $z_2$ and $z_3$,
which are adjacent to $x$ or $y$.
Since $k_i(z_i)=2$, where $1\le i\le 3$, $\kappa_Q\cup\{k(z_1),k(z_2),k(z_3)\}$
contains the desired sequence $\kappa$.
Since $k_7(G) = 2 < k_7(H_{n,2,2})=3$, we conclude that $W(G)<W(H_{n,2,2})$.
\end{proof}

For the sake of completeness, in Table~2 we present Wiener indices of
$C_n$, $H_{n,1,2}$, $H_{n,2,2}$ and $H_{n,1,3}$.
These indices are calculated using the fact that all the considered graphs
have one large isometric cycle of length $t\le n$ plus some extra vertices.
Notice that ${\bf 2}_t=\frac{t^2-1}4$ if $t$ is odd and
${\bf 2}_t=\frac{t^2}4$ if $t$ is even.

\begin{center}
\begin{tabular}{c|c|c|}
& odd $n$ & even $n$ \\
\hline
$W(C_n)$ & $\frac 18(n^3-n)$ & $\frac 18(n^3)$ \\
$W(H_{n,1,2})$ & $\frac 18(n^3-n^2+3n-3)$ & $\frac 18(n^3-n^2+2n)$ \\
$W(H_{n,2,2})$ & $\frac 18(n^3-n^2-n+17)$ & $\frac 18(n^3-n^2-2n+16)$ \\
$W(H_{n,1,3})$ & $\frac 18(n^3-2n^2+11n-18)$ & $\frac 18(n^3-2n^2+12n-16)$ \\
\hline
\end{tabular}
\end{center}
\begin{center}
Table~2: Wiener indices of $C_n$, $H_{n,1,2}$, $H_{n,2,2}$ and $H_{n,1,3}$.
\end{center}

Let $H_{n,2,2}^+$ be the graph obtained from $H_{n,2,2}$ by adding an edge between two vertices that are at distance 1 from the vertices of degree 3.
As a remark, we note that $H_{n,2,2}^+$ has Wiener index exactly $W(H_{n,2,2}) - 1$, so it is the (possibly not unique) fourth $2$-connected graph by decreasing Wiener index for $n = 9$ and $n \ge 11$. We conjecture that for $n$ large enough, it is the unique fourth $2$-connected graph by decreasing Wiener index, and that the unique fifth such graph is $H_{n,1,3}$.

\vskip 1pc
\noindent{\bf Acknowledgements.}~~The third author acknowledges
partial support by Slovak research grants APVV-15-0220, APVV-17-0428,
VEGA 1/0142/17 and VEGA 1/0238/19.
The research was partially supported by Slovenian research agency ARRS,
program no. P1-0383.

\end{document}